\documentclass[]{article}

\usepackage{graphicx}
\usepackage{subfigure}

\usepackage{amssymb}

\usepackage{amsthm}
\usepackage{mathrsfs}
\usepackage{amsmath}

\usepackage{array}
\usepackage{booktabs}

\newtheorem{theorem}{Theorem}[section]
\newtheorem{lemma}[theorem]{Lemma}
\newtheorem{proposition}[theorem]{Proposition}
\newtheorem{corollary}[theorem]{Corollary}

\newcommand{\tabincell}[2]{\begin{tabular}{@{}#1@{}}#2\end{tabular}}

\title{Routing for Energy Minimization with Discrete Cost Functions}

\begin{document}




\author{\small Lin Wang$^{1,3}$, Antonio Fern\'andez Anta$^2$, Fa Zhang$^1$, Chenying Hou$^{1,3}$, Zhiyong Liu$^{1,4}$\\
\small $^1$ Institute of Computing Technology, Chinese Academy of Sciences \\
\small $^2$ Institute IMDEA Networks\\
\small $^3$ University of Chinese Academy of Sciences\\
\small $^4$ China Key Lab for Computer Architecture, ICT, CAS}
\date{}

\maketitle





\begin{abstract}
Energy saving is becoming an important issue in the design and use of computer networks. In this work we propose a problem that
considers the use of rate adaptation as the energy saving strategy in networks. The problem is modeled as an integral demand-routing 
problem in a network with discrete cost functions at the links. The discreteness of the cost function comes from the different states (bandwidths) at which links can operate and, in particular, from the energy consumed at each state. This in its turn leads to the non-convexity of the cost function, and thus adds complexity to solve this problem. We formulate this routing problem as an integer program, and we show that the general case of this problem is NP-hard, and even hard to approximate. For the special case when the step ratio of the cost function is bounded, we show that effective approximations can be obtained. Our main algorithm executes two processes in sequence: relaxation and rounding. The relaxation process eliminates the non-convexity of the cost function, so that the problem is transformed into a fractional convex program solvable in polynomial time. 
After that, a randomized rounding process is used to get a feasible solution for the original problem. This algorithm provides a constant approximation ratio for uniform demands and an approximation ratio of $O(\log^{\beta-1} d)$ for non-uniform demands, where $\beta$ is a constant and $d$ is the largest demand. 

\end{abstract}





\section{Introduction}
\label{sec:intro}

Energy-aware computing has recently become a hot research topic. The increasingly widespread use of Internet and the sprouting of data centers are having a dramatic impact on the global energy consumption. The energy consumed comes from the aggregate power used by many devices (CPUs, hubs, switches, routers). As shown in \cite{Greenberg_Hamilton-2009}, up to $45\%$ of the total energy is consumed by computing resources such as CPUs and storage systems. There has been significant amount of work done on energy management schemes for these largest energy consumers, which was started by the work on power consumption minimization in processors done by Yao et al. \cite{Yao_Demers-1995}. In parallel with energy saving in computing devices, the energy consumption of network devices has also been a fundamental concern in wired networks like Internet or data center networks. Gupta and Singh \cite{Gupta_Singh-2003} were among the first to identify energy consumption of general networks as an important issue. Kurp et al. \cite{Kurp-2008} then showed that there is significant room for saving energy in current networks in general. The main reason is that networks are always designed with a significant level of redundancy and over-provisioning, in order to guarantee QoS, and to tolerate peak load and traffic variations. However, since networks usually carry traffic that is only a small fraction of the peak, a significant portion of the energy consumed is wasted. This is particularly representative in data center networks, where the connectivity redundancy is very heavy and the traffic load is dramatically fluctuant with certain patterns over time. Ideally, the energy consumed in a network should be proportional to the traffic load it carries.

In this paper we consider a centralized energy saving problem for networks. In principle, this problem has applications in most networks, and can particularly be used in data center networks or other software-defined networks. In the model assumed, we use an energy saving strategy called \emph{rate adaptation} , which 
assumes that the network links have a discrete set of operational states (bandwidths), each with its corresponding power consumption. 
By using the rate adaptation strategy, network devices choose an appropriate speed according to their current traffic load. 
This model is different from the usual speed scaling model (see below) that assumes a continuous power function, and is more realistic. However, the optimization problem we need to solve in a model with rate adaptation is more complex compared with using these other models. Gunaratne et al. \cite{Gunaratne_Christensen-2008} proposed a network model in which links use rate adaptation, that they called adaptive link rate (ALR). We will present a formal description of our network model, that uses rate adaptation as power saving strategy, and of the problem of providing route assignment to demands from a global view of the network. Actually, this model turns out to be a network routing problem with a discrete cost function. We show by analysis that in the general case, this problem is hard to approximate. Given some restriction on the cost function, we provide an efficient approximation algorithm for it.

\subsection{Related Work}

Work on energy efficiency have mostly focused on two strategies to save energy: speed scaling and powering down. Under speed scaling, it is assumed that the power consumed by a device working at speed $s$ has the form of $P=s^{\beta}$, where $\beta$ is a constant between one and three. This comes from the well known cube-root rule, which states that the speed is approximately the cube root of the power consumed. Speed scaling  has been widely used for optimizing the energy consumption of computing resources, especially the processor (see \cite{Yao_Demers-1995, Chan_Chan-2007,Bansal_Kimbrel-2007,Bansal_Kimbrel-2009}). This strategy has also been used for network devices. However, most studies in this category focused on a single device. At the same time, speed scaling as described does not model realistically network devices and is hard to be applied in practice. Another approach to save energy is achieved by powering down the devices while they are idle. Andrews et al. \cite{Andrews_Fernandez-PD-2010} considered that network elements operate only in the full-rate active mode or the zero-rate sleeping mode. Nedevschi et al. explored both speed scaling and power down to reduce global network energy consumption. \emph{ElasticTree} \cite{Heller_Seetharaman-2010} is a centralized energy controlling system for data center networks which turns off some of the routers or switches and then yields energy-efficient routes.

The general network routing problem is described as follows. We are given a set of traffic demands and want to unsplitably route them over a transmission network. The total traffic $x_{e}$ on link $e$ incurs a cost which is defined by a cost function $f_{e}(x_{e})$. Our objective is to find routes for all demands such that the total incurred cost $\sum_{e}f_{e}(x_{e})$ is minimized.
There has been significant work on this general network routing problem which is also called the minimum-cost multi-commodity flow (MCMCF) problem. Note that the complexity of this problem depends on the cost function defined on each edge. In the simplest case where the cost function $f_e(x_e)$ is linear with the carried traffic $x_e$, this problem turns out to be the shortest path problem which is polynomial-time solvable. Moreover, if we choose $f_{e}(x_e)$ as a subadditive continuous function which has the property of \textit{economies of scale}, the problem becomes the Buy-at-Bulk (BAB) problem which has been well studied. Awerbuch and Azar \cite{Awerbuch_Azar-1997} provided an $O(\log^{2}n)$ randomized approximation algorithm for this problem. Andrews \cite{Andrews-2004} showed that for any constant $\gamma > 0$, there is no $O(\log^{\frac{1}{4}-\gamma}n)$-approximation algorithm for uniform BAB (all edges have a uniform cost function), and there is no $O(\log^{\frac{1}{2}-\gamma}n)$-approximation algorithm for the non-uniform version, unless $\mbox{NP}\in\mbox{ZPTIME}(n^{polylog\; n})$. Here $n$ is the size of the network. However, for general cases with superadditive continuous cost functions, achieving any finite ratio approximation is NP-hard \cite{Andrews_Fernandez-SS-2010}. 
Andrews et al.~\cite{Andrews_Fernandez-SS-2010} also studied an energy-aware scheduling and routing problem for wired networks under the speed scaling model,
presenting a constant approximation algorithm for uniform demands. 
In \cite{Andrews_Antonakopoulos-2010,Andrews_Antonakopoulos-2011} the authors studied routing under a special case of speed scaling with a startup cost, and presented a polylogarithmic approximation algorithm. The most significant difference between the work mentioned above and our work is that here we focus on the network routing problem with discrete cost functions rather than continuous ones.

\subsection{Our Results}

Based on the general network routing problem, our problem can be described as follows. We want to route a set of traffic demands over a transmission network. Each link (which abstracts both the transmission link and the corresponding ports of the end devices connected by this link) can operate at one of a given set of speed states $R_1< \cdots < R_m$. For each speed state $R_i$, a fixed amount of cost $E_i$ is defined to represent the power consumption at this speed. To carry all traffic flows, each link $e$ chooses a speed state $R_i$ at which to operate so that $x_e \in (R_{i-1},R_i]$. If $x_e \le R_1$, then $R_1$ is chosen. We also assume that $x_e \le R_m$ for any routing considered. Consequently, each link will incur a cost and the total cost of the network can be defined as $\sum_{e} E^e_i$, where $E^e_i$ is the cost of link $e$ while choosing a speed $R_i$. Equally we can use a cost function $f_e(x_e)$ to describe the cost on each edge. This cost function can be obtained by setting the value of $f_e(x_e)$ to $E_i$ if $x_e \in (R_{i-1}, R_i]$. Thus the total cost can be represented as $\sum_e f_e(x_e)$. It is trivial that function $f_{e}(x_{e})$ has a discrete formulation and is indeed a step function. For simplicity, we focus on the uniform case, where all the links consume energy in the same way and thus have the same cost function $f(\cdot)$. 

We aim to solve the minimum-energy routing with this discrete cost function in this paper. In Section~\ref{sec:model}, we give the formalized description of the model and prove that in general we cannot get any finite ratio polynomial-time approximation algorithm for it if P$\neq$NP. In Section~\ref{sec:approx}, we consider a special case where $f(1)$ and the ratio between any two adjacent steps of the cost function $f(\cdot)$ is bounded by a constant. Given this assumption, we show that efficient approximation algorithms can be achieved. The approximation we provide in this paper is established on a two-step relaxation and rounding process. The first relaxation is made on the cost function to eliminate the discreteness, which is done by a special interpolation method which transforms the discrete function $f(\cdot)$ into a continuous one $g(\cdot)$ while introducing a bounded error. Then, by relaxing the integral indicator variable, the induced problem turns out to be a convex program and can be solved in polynomial time. Hereafter, a two-step rounding technique is used to round infeasible solutions to the feasible region. We show by analysis in Section~\ref{sec:analysis} that our approach can achieve a constant approximation for the cases with uniform demands and an $O(\log^{\beta-1}d)$-approximation to the cases with non-uniform demands, where $d$ is the maximum one among all traffic demands. 
Finally in Section~\ref{sec:con}, we draw conclusions and provide some future work and open problems derived from this work.

\section{The Model}
\label{sec:model}

We study the energy-efficient routing problem under a general case where the optimizer takes a traffic matrix and a network topology as inputs and returns the network configuration. We first consider a simple case with unit demands where all the demands have a flow amount of one integer unit, and then extend the results to the uniform and non-uniform cases, where we consider that all the traffic demands have the same flow amount of some integer units and have different flow amounts, respectively.

The problem is described as follows. Assume we are given an undirected graph $G=(V,E)$ and a set of traffic demands $D=\{d_1, d_2, ..., d_k\}$, where the $i^{th}$ demand, $1 \le i \le k$, requests $d_i$ units of bandwidth provisioned between a source node $s_i$ and a sink node $t_i$. Unless otherwise stated, in the following we assume unit demands, i.e., $d_i=1$. We also assume that links represent abstract resources (links with their switch ports), and each link can operate at one of a constant number of different rates $R_1 < R_2 < ... < R_m$. Here we assume that $R_1 \ge 1$ because a link will carry a flow of at least one unit. For energy-saving consideration as we mentioned in the previous section, it is reasonable to have many different rates in future network devices. Each rate $R_i$, $1 \le i \le m$, has a cost of $f_e(R_i)$ which is the power consumption. Our goal is to route every demand in an unsplittable fashion with the objective of minimizing the total cost $\sum_e f_e(R_e)$, where $R_e \in \{R_i | i \in [1, ..., m]\} $. Note that unsplittable routing is important in networks in order to avoid packet reordering. We only consider the case in which all links have the same cost function $f(\cdot)$. We call this energy-efficient routing problem with discrete cost functions as \emph{rate adaptive energy-efficient routing problem}.

Not surprisingly, the rate adaptive energy-efficient routing problem is NP-hard, due to the fact that the cost function is discrete. Furthermore, we show that, in general, it cannot even be approximated. This is shown by the following theorem.

\begin{theorem}
\label{thm:inapprox}
For any constant $\rho \ge 1$, there is no polynomial-time $\rho$-approximation algorithm for the rate adaptive energy-efficient routing problem, unless P=NP.
\end{theorem}

\begin{proof}

We show here that a polynomial-time $\rho$-approximation algorithm $\mathscr{A}$ could be used to
solve the edge-disjoint paths (EDP) problem, which is NP-hard, proving the claim. Let us assume the algorithm $\mathscr{A}$ exists.
Let us consider an instance of the EDP problem, given as a network with $w$ links and a collection of $k$ pairs source-sink, $(s_j,t_j)$.
We transform it into an instance of the rate adaptive energy-efficient routing problem as follows. In this instance, the network, the number of demands, and the
source $s_j$ and sink $t_j$ of each demand $j$ are the same as in the EDP instance. Also, there are only two rates $R_1$ and $R_2$, $R_2 \geq k \cdot R_1$,
and we choose any cost function such that $f(R_2) > \rho \cdot w \cdot f(R_1)$. Finally, we set each demand value $d_ j=R_1$. This transformation is
trivially polynomial in time.

We use the algorithm $\mathscr{A}$ to find a solution that $\rho$-approximates the optimal solution of the rate adaptive energy-efficient routing problem
obtained by the transformation. (Observe that since $R_2 \geq k \cdot R_1$, a solution always exist.)
Let $\hat{C}$ be the cost of the solution found. We claim that $\hat{C} \leq \rho \cdot w \cdot f(R_1)$ if and only if
there are edge disjoint paths for all the demands. If this claim holds, then $\mathscr{A}$ can be used to solve EDP in polynomial time and hence $P=NP$.

The claim can be proven as follows. 
($\Rightarrow$) Assume that $\hat{C} \leq \rho \cdot w \cdot f(R_1)$. Then, the paths found by $\mathscr{A}$ are edge disjoint. Otherwise, at least one of the links would be used by at least two demands, the rate of that link would be $R_2$, and $\hat{C} \geq f(R_2) > \rho \cdot w \cdot f(R_1)$.
($\Leftarrow$)
Assume edge disjoint paths exist for all demands. Then, there is a way to route all the demands so that
each link is used by no more than one demand, and can be set to rate $R_1$. The cost of every edge in this solution is at most $f(R_1)$, and the cost of the solution is hence $C \leq w \cdot f(R_1)$. Trivially, the optimal solution also has cost $C^* \leq w \cdot f(R_1)$.
Since $\mathscr{A}$ is a $\rho$-approximation algorithm, the cost of the solution it outputs must satisfy 
$\hat{C} \leq \rho \cdot C^* \leq \rho \cdot w \cdot f(R_1)$.
\end{proof}

Actually Lemma 2 in \cite{Andrews_Fernandez-SS-2010} is a special case of this theorem, in which in the proof $R_1=1$ and $f(R_1)=0$. 
The proof of the theorem suggests that we cannot get any finite-ratio polynomial-time approximation algorithm for the rate adaptive energy-efficient routing problem unless we bound the ratio of costs of different rates $f(R_i)/f(R_{i-1})$ and $f(1)/\mu$ (we call it as step ratio). In the rest of this work, we will consider that this ratio is bounded by a constant $\sigma$. Note that this is quite sensible because the energy consumption of devices is always bounded. But, even in this special case this problem is NP-hard. The good news are that the problem with constant step ratio can be approximated by a constant approximation ratio. We will show the details in the following sections.

Formally, we formulate the described routing problem into an integer program, that can be seen as $(P_{1})$ below. The binary variable $y_{i,e}$ indicates whether demand $i$ uses link $e$, while $x_{e}$ is the total load on $e$. Flow conservation means that for each demand $i$ the source $s_{i}$ generates one unit of flow, the sink absorbs one unit of flow, and for the other vertices the incoming and outgoing flows of demand $i$ are the same. Observe that for $x_{e}\leq z_{e}$, $f(x_{e})=f(z_{e})$. This results in the discrete property of the cost function $f(\cdot)$. More precisely, $f(x)$ is a non-decreasing step function of $x$, where $x$ is the speed of each link. In practice, cost functions for network resources might different, but here we focus on the case with only a uniform cost function. There is no doubt that solving $(P_1)$ is NP-hard for the $0-1$ constraint on variable $y_{i,e}$. Since solving our network routing problem is NP-hard, we have no hope on finding the optimal solution.

\begin{align*}
 & (P_{1})\;\; \min\;\;\sum_{e}f(z_{e})\\
\mbox{subject to}\\
 & x_{e}=\sum_{i}y_{i,e} & \forall e\\
 & x_{e}\leq z_{e} & \forall e\\
 & z_{e}\in\lbrace R_{1},...,R_{m}\rbrace & \forall e\\
 & y_{i,e}\in\lbrace0,1\rbrace & \forall i,e\\
 & y_{i,e}:\;\; \mbox{flow conservation}
\end{align*}

\section{Approximation}
\label{sec:approx}

In this section, we present our main algorithm, which is devised to approximate the optimal solution of $(P_1)$. As indicated above, the complexity of the rate adaptive energy-efficient routing problem comes from the non-convexity of the cost function $f(\cdot)$ and the integral constraint on $y_{i,e}$. Intuitively, if we remove these constraints and transform function $f(\cdot)$ into a convex one, the induced problem could be solved. From this intuition, we propose an approximation algorithm for this problem. By a two-step relaxation and rounding procedure, the algorithm achieves a constant approximation.

In order to get an approximated solution of $(P_1)$, a two-step relaxation and rounding procedure is introduced. The first relaxation is made on the cost function $f(\cdot)$, to remove the discreteness of $f(\cdot)$. We use a particular interpolation method to transform the discrete cost function of the original program into a continuous convex one. It makes the program to be simpler while introducing a bounded error. Additionally, we relax the binary variable $y_{i,e}$ to be real. Then, the induced problem is a convex program which has been proved to be optimally solvable in polynomial time. Now we describe our algorithm in detail.

\subsection{Relaxing the Cost Function}

The first relaxation is made on the cost function. We use a special interpolation method to simplify ($P_1$) by replacing the step function $f(\cdot)$ with a continuous function $g(\cdot)$. 

Before applying the interpolation, it is important to decide the form of the function $g(\cdot)$ which we aim to obtain. Due to the well known cube-root rule, it has been suggested that most network devices consume energy in a superadditive manner (\cite{Andrews_Fernandez-SS-2010, Leon_Navarro-2011}). That is, doubling the speed more than doubles the energy consumption. In particular the energy curve is often modeled by a polynomial function $g(x)=\mu x^\beta$, where $\mu$ and $\beta$ are constants associated to the network elements. More precisely, the parameter $\beta$ in the ordinary form of energy consumption has been usually assumed to be in the interval $(1,3]$ \cite{Brooks_Bose-2000}. The objective here is to transform a step cost function $f(\cdot)$ into a function in the form of $g(x)=\mu x^\beta$. Although, as mentioned, $\beta$ will be typically larger than $1$, and hence $g(\cdot)$ will be a convex function, the proposed interpolation method does not impose any restriction on this.

Now we discuss how to apply the transformation from a step function to a continuous (convex) one. A common approach is to use the midpoints of all the steps as discrete values and fit by mean squares. This approach will not give good results if the steps of the function have different lengths, because there might be many steps with small lengths dominating the interpolation. Another popular method is to perform an interpolation on a set of points which is obtained by sampling the original function. Unfortunately, using this technique the error of the interpolation depends on the sampling method we choose, and is always hard to estimate. In order to contain the interpolation error, we devise a new interpolation method which based on minimizing the difference between the two function. Without depending on some other parameters, the method works well especially for the fitting of step functions.

Consider the original function $f(x)$, and the one to be fitted $g(x)$, as previously mentioned. $f(x)$ is defined as follows.

\begin{equation}
\label{eq:f_def}
f(x)=
\left\lbrace
\begin{aligned}
	& y_1, & 1 \le x \le R_1, \\
	& y_2, & R_1 < x \le R_2, \\
	& ... \\
	& y_m, & R_{m-1} < x \le R_m,
\end{aligned}
\right.
\end{equation}
where in our case $y_i=f(R_i)$, $(1 \le i \le m)$ is the cost of each state and $R_i$, $R_{i+1}$ $(1 \le i < m)$ represent the lower and upper boundaries of the speed for each state. We aim to fit $g(x)$ to $f(x)$.

The formula that has to be minimized can be represented as

\begin{equation}
\label{org_formula}
G(\mu, \beta) =  \max_{x \in [1, R_m]} \left\lbrace \frac{f(x)}{g(x)}, \frac{g(x)}{f(x)}\right\rbrace.
\end{equation}
We only consider $x \in [1, R_m]$ because $f(0) = g(0)$ and in any feasible integral solution, $x \not \in (0,1)$. Since $g(x)$ is not a linear function, this minimization problem is hard to solve. We then consider an alternative. We use the fact that the formula becomes linear with parameters $\mu$ and $\beta$ when a logarithmic transformation is applied. The fitting function $g(x)$ is clearly linear under a logarithmic transformation. Note that

\begin{equation}
\log (g(x)) = \log \mu + \beta \log x.
\end{equation}
Then it is equivalent if we minimize the following formula
\begin{equation}
\label{log_formula}
\begin{aligned}
	G'(\mu, \beta) = &  \max_{x \in [1, R_m]} |\log f(x) - \log g(x)| \\
					= & \max_{x \in [1, R_m]} |\log f(x) - (\log \mu + \beta \log x)|.
\end{aligned}
\end{equation}
However, the above formula of $G'$ is still unable to tackle, because the absolute operation is hard to handle. Here we propose to use an alternative which is to minimize the integral of the square of the difference between the two functions. Let us define $v_i=\log y_i$, $w_i=\log R_i$ for $i \in [1,m]$ and $w_0 = \log 1 = 0$, and $\mu'=\log \mu$. Then the alternative formula that we will in fact use is

\begin{equation}
\label{eq:log_integral}
H(\mu', \beta) = \sum_{i=1}^m \int_{w_{i-1}}^{w_{i}} [v_i - (\mu' + \beta w)]^2 dw.
\end{equation}
And now Eq.~(\ref{eq:log_integral}) is to be minimized with respect to the parameters of the general quadratic equation. We choose $\mu'$ and $\beta$ such that the first partial derivatives of $H(\mu', \beta)$ are equal to zero and its second derivatives are positive.

\begin{equation}
\label{eq:derivative}
\left\lbrace
\begin{aligned}
	& \frac{\partial H}{\partial \mu'} = \sum_{i=1}^m \int_{w_{i-1}}^{w_{i}} -2[v_i-\mu' - \beta w]dw = 0, \\
	& \frac{\partial H}{\partial \beta} = \sum_{i=1}^m \int_{w_{i-1}}^{w_{i}} -2w[v_i - \mu' - \beta w]dw = 0.
\end{aligned}
\right.
\end{equation}
It is obvious that the second derivatives are all positive. By solving Eq.~(\ref{eq:derivative}), we get the values of parameters $\mu'$ and $\beta$, and from $\mu'$ we obtain $\mu$. Consequently, the objective function $g(x)$ of the interpolation can be determined.
Once the function $g(x)$ is obtained, the optimization problem can be rewritten as follows.
\begin{align*}
 & (P_{2})\;\; \min\;\;\sum_{e}g(x_{e})\\
\mbox{subject to}\\
 & x_{e}=\sum_{i}y_{i,e}& \forall e\\
 & y_{i,e}\in\lbrace0,1\rbrace & \forall i,e\\
 & y_{i,e}:\;\; \mbox{flow conservation}
\end{align*}
The problem now turns to be an integer program with a convex objective function. Solving $(P_{2})$ is still NP-hard.

\subsection{Relaxing the Binary Constraint}

In this section, we show that by relaxing the binary constraint on $y_{i,e}$, $P_{2}$ can be efficiently solved. Although this brings some accuracy loss, it makes the problem solvable. It is feasible to do this kind of relaxation in our case, with some small modifications. One key observation is the convexity of the objective function in $P_{2}$. By combining the linear constraints and relaxing the binary variable $y_{i,e}$ from $\{ 0,1 \}$ to $[0,1]$, $P_{2}$ turns to be a regular convex program and can be solved in polynomial time by using algorithms for convex programming. We denote the fractional solution obtained by convex programming as $y_{i,e}^*$. However the problem is that in the fractional solution $y_{i,e}^*$, the traffic flow of a demand can be splitted over multiple paths, which is not feasible in our case. We address this below.

\subsection{Two-step Rounding}

In this section, we introduce a two-step rounding technique to transform the fractional solution into a feasible one. First we round the fractional routing solution to an integral one, and then we determine the link rates from the routes.

For the routing path rounding, the key point is how to transform the fractional routes $y_{i,e}^*$ into a solution in which there is only one path for each traffic flow. In the fractional solution, the flow $i$ carried by link $e$ is represented as $y_{i,e}^*$, where $y_{i,e}^* \in [0,1]$. Our goal is to get a path $p_i$ for each demand $i$ which satisfies $\{ e | e\in p_i \} \subseteq \{ e | y_{i,e}^* > 0 \}$, to route flow $i$. We use the Raghavan-Thompson randomized rounding method to complete this transformation.

The overall rounding process is described below. Once the optimal fractional solution $y_{i,e}^*$ has been found, the flows assigned to the links is mapped to paths as follows. For each demand $i$, first we generate a sub-graph $G_{i}$ defined by links $e$ where $y_{i,e}^* > 0$. Then, we extract a simple path $p$ connecting the source and destination nodes. The link $e \in p$ with the smallest weight $y_{i,e}^*$ is called the bottleneck link. This $y_{i,e}^*$ is selected as the weight of this path, denoted as $w_{p}$. Hereafter the weight $y_{i,e}^*$ of each link in path $p$ is decreased by $w_{p}$. The above procedure will be repeated until that for all $e$ we have $y_{i,e}^* = 0$. By the flow conservation constraint, we can state that

\begin{proposition}
By repeating the path extraction process on $G_i$, a state in which $\forall e \in G_i$, we have $y_{i,e}^* = 0$ will be reached.
\end{proposition}
Consequently, we define this state as the termination of the path extraction process. As a result we get a collection of paths $\{p\}$ and for each of them we have a weight $w_p$. We randomly select one path from $\{p\}$ for current demand $i$ using the path weights as the selection probabilities. After this rounding, there will be only one single path for each demand. We denote this solution as $\hat{y_{i,e}} \in \{0,1\}$.

After the routes for the demands have been chosen, the state of every link must be determined. We select the speed of each link via the following rounding procedure. First, we compute the carried traffic $\hat{x_e}=\sum_i \hat{y_{i,e}}$, where $\hat{y_{i,e}}$ is the amount of demand $i$ that traverses link $e$ after the rounding. Then for each link, we search the collection of possible operational speeds and choose the minimal $s_e$ which can support the carried traffic. Formally we have

\begin{equation}
\label{eq:rounding-2}
s_e = \min \{ R_i | (i \in [1,m]) \wedge (\hat{x_e} \le R_i) \}.
\end{equation}

Determining the link states and routing all the traffic demands result in an approximated solution for the rate adaptive energy-efficient routing problem with discrete cost functions.

\section{Algorithm Analysis}
\label{sec:analysis}

We evaluate our method in this section. As it is designed to approximate the optimal, it is essential to derive an approximation ratio as the main accuracy criterion. Some accuracy loss has been introduced in both relaxation and rounding processes, so we study them one by one. First, we show a bound on the error introduced by the interpolation. Second, we bound the error by rounding. Then, by combining these results together, we draw our main conclusions. 

\subsection{Bound on the Interpolation Error}

Interpolation is proposed to approximate the objective function, so it is important to bound the error introduced. During the interpolation process, the error comes from the gap between the original function $f(x)$ and the fitted function $g(x)$. We define this gap as follows:

\begin{equation}
gap=\max_{x \in [1, R_m]} \left\lbrace \frac{f(x)}{g(x)}, \frac{g(x)}{f(x)} \right\rbrace.
\end{equation}
Here we also consider $x \in [1,R_m]$ as we did in previous sections. We will use this gap definition as the \emph{interpolation error}, reason for which we present the following theorem. Here we assume that $g(x)$ intersects with $f(x)$ in each step of $f(x)$, which is reasonable because we assumed that $g(x)$ could well describe $f(x)$ under the cube root rule. If not, there could be a big difference in trend between $f(x)$ and $g(x)$, and we cannot obtain any bounds. In other words, this would mean that the cube root rule does not apply to this network.

\begin{theorem}
\label{thm:error_bound}
Given function $f(x)$, if $f(x)$ satisfies $y_i/y_{i-1} \le \sigma$ where $\sigma > 1$, then in interval $[1, R_m]$, the interpolation error satisfies $gap \in [\frac{2\sigma}{\sigma + 1}, \varphi]$, where $\varphi = \max \left\{\sigma, f(1)/\mu\right\}$.
\end{theorem}

\begin{proof}
First we show the lower bound. 
Assuming $f(x)$ and $g(x)$ intersecting at each gap, and given the form of $g(\cdot)$, means that we can focus on values of $x \in \{1, R_1, R_2,..., R_m\}$. Hence we consider a particular point $x$ from the set $\{1, R_1, R_2,..., R_m\}$  (see Fig.~\ref{fig:error_bound}).
Assume that we can bound the gap and that we have a bound $\delta$. Let $\gamma = \max_i \{ y_i/y_{i-1} \}$.

Case 1: The fitted function $g(x)$ is above the midpoint of the two values of the step function $f(x)$ at $x$, as can be seen in Fig.~\ref{fig:error_bound:a}. Therefore, at $x$, the gap is defined as $gap=g(x)/f(x)$. Suppose we are given $\delta < \frac{2\gamma}{\gamma+1}$. Then for $\gamma>0$, we have
\begin{equation}
\delta < \frac{2\gamma}{\gamma+1} < \frac{\gamma+1}{2}.
\end{equation}
Then we can obtain
\begin{equation}
\gamma = \frac{y_{i}}{y_{i-1}} > (2\delta-1).
\end{equation}
Notice that $y_{1}$ equals $f(x)$ and here $g(x)$ is not smaller than $y_{i-1}+\frac{y_{i}-y_{i-1}}{2}$. We have
\begin{equation}
\begin{aligned}
	g(x) & \ge y_{i-1} + \frac{y_{i}-y_{i-1}}{2} \\
				   & \delta y_{i-1} = \delta f(x).
\end{aligned}
\end{equation}
This results in $g(x)>\delta f(x)$, which is a contradiction to the assumption that $g(x)/f(x) \le \delta$. So we have $\delta>\frac{2\gamma}{\gamma+1}$ in this case.

\begin{figure}[!t]
  \centering
  \subfigure[$g(x) \ge \frac{y_i + y_{i-1}}{2}$]{
    \label{fig:error_bound:a} 
    \includegraphics[width=2in]{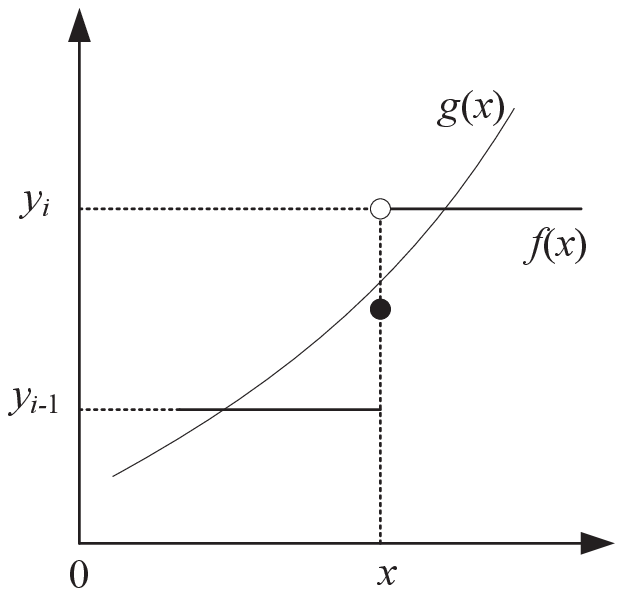}}
  \hspace{0.1in}
  \subfigure[$g(x) < \frac{y_i + y_{i-1}}{2}$]{
    \label{fig:error_bound:b} 
    \includegraphics[width=2in]{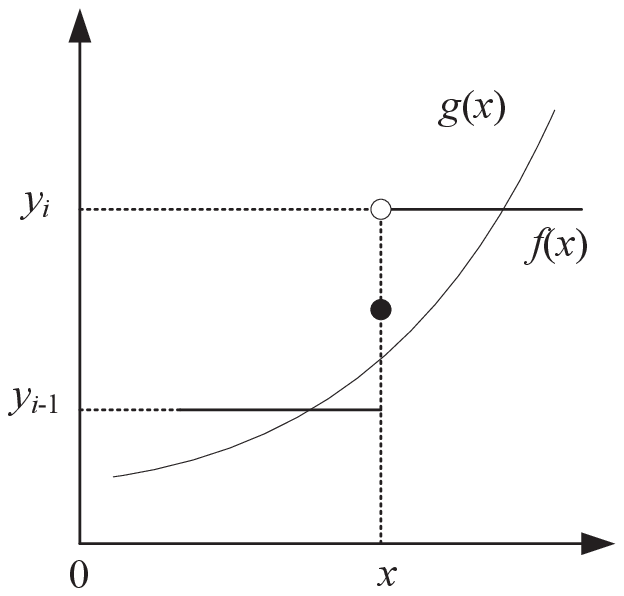}}
  \caption{Two cases of error bounding proof}
  \label{fig:error_bound} 
\end{figure}

Case 2: The fitted function $g(x)$ is below the midpoint of the two values of $f(x)$ at $x' = x + \epsilon$ where $\epsilon$ is infinitesimal, as can be seen in Fig.~\ref{fig:error_bound:b}. Here the gap is defined as $gap=f(x')/g(x')$. Suppose we are given $\delta < \frac{2\gamma}{\gamma+1}$. We have

\begin{equation}
2y_{i} > \delta(y_{i-1}+y_{i}).
\end{equation}
Note that $y_{i}=f(x')$ and $g(x') \le y_{i-1} + \frac{y_{i}-y_{i-1}}{i}$. Then we can derive

\begin{equation}
\begin{aligned}
	g(x') & \le y_{i-1} + \frac{y_{i}-y_{i-1}}{2} \\
				   & \le \frac{y_{i}}{\delta} = \frac{i-1}{\delta}f(x').
\end{aligned}
\end{equation}
Similarly to the first case, the inequality $f(x') > \delta g(x')$ results in a contradiction to the assumption that we have a bound. So we also get $\delta>\frac{2\gamma}{\gamma+1}$ in this case.

As we have obtained the lower bound, now we discuss about the upper bound. It is straightforward that we cannot get a tighter upper bound rather than $\delta$ for the error at any point $x \in [R_1,R_2,...,R_m]$. For the extreme case where the fitted function $g(x)$ takes a value which is very close to $y_{i}$ at $x$, the gap can be represented by $gap=g(x)/f(x) \approx y_{i}/y_{i-1}$, as well as the error. However, when $x \in [1, R_1)$, this gap may be bounded by the case when $x = 1$, where $gap = f(1)/\mu$. Therefore, the interpolation error can be upper bounded by $\varphi =  \max \left\{\sigma, f(1)/\mu\right\}$.

Combining these results, we have that the interpolation error is in $\left[ \frac{2\gamma}{\gamma+1}, \varphi \right]$. This completes the proof.
\end{proof}

\begin{table}[!t]
	\caption{\label{tab:notation} Notation for Bounding the Approximation Ratio}
	\centering
	\begin{tabular}{cl}
		\toprule
		Name & Definition \\
		\midrule
		$C_f$ & optimal integral solution under $f(\cdot)$\\
		\midrule
		$C_{g}$ & optimal integral solution under $g(\cdot)$ \\
		\midrule
		$C_{f}^*$ & \tabincell{l}{optimal fractional solution obtained by the program \\with a binary relaxation, which is $\sum_e f(x_e^*)$} \\
		\midrule
		$C_g^*$ & \tabincell{l}{optimal fractional solution under $g(\cdot)$, which is $\sum_e g(x_e^*)$} \\
		\midrule
		$\hat{C_g}$ & \tabincell{l}{solution obtained after the first rounding process, which is $\sum_e g(\hat{x_e})$} \\
		\midrule
		$\hat{C_f}$ & solution obtained by our overall method \\
		\bottomrule
	\end{tabular}
\end{table}

\subsection{Approximation Ratio Analysis}

In this section, we present our results for the approximation ratio. The notation that we are going to use are presented in Table~\ref{tab:notation}. Given this notation, our objective is to bound the ratio between $\hat{C_f}$ and $C^f$. 
We will assume that $\sigma$ and $\varphi$ are bounded by constants and will show that the ratio is then bounded by a constant.

\begin{lemma}[\cite{Andrews_Fernandez-SS-2010}]
\label{lm:bound_rounding-1}
For unit demands, randomized rounding on paths can obtain a $\lambda$-approximation for the energy-efficient routing problem such that $E[\hat{C_g}] \le \lambda E[C_g^*]$, where $\lambda$ is a constant.
\end{lemma}

Since $C_g^*$ is obtained by the optimal fractional solution and it is always the lower bound of the optimal integral solution $C_g$, we directly have that $E[\hat{C_g}] \le \lambda E[C_g^*] \le \lambda E[C_g]$.

\begin{lemma}
\label{lm:bound_rounding-2}
If the ratio between any two adjacent steps of cost function $f(\cdot)$ is bounded by $\sigma$, then $\hat{C_f}\le \varphi \sigma \hat{C_g}$.
\end{lemma}

\begin{proof}
The result follows from Theorem~\ref{thm:error_bound}, in which it shows that the largest gap between $g(\hat{x_e})$ and $f(\hat{x_e})$ is $\varphi$. Then, from the relation between $s_e$ and $\hat{x_e}$ (see Eq.~(\ref{eq:rounding-2})), we have that $f(s_e) \le \sigma f(\hat{x_e})$. By Theorem~\ref{thm:error_bound}, $f(s_e) \le \varphi \sigma g(\hat{x_e})$ follows. Thus, the following result can be derived
\begin{equation}
\hat{C_f} = \sum_e f(s_e) \le \varphi \sigma \sum_e g(\hat{x_e}) = \varphi \sigma \hat{C_g}.
\end{equation}
This completes the proof.
\end{proof}

\begin{theorem}
For unit demands, the expected energy consumption $\hat{C_f}$, is a $\rho$-approximation of $E[C_f]$. That is, $E[\hat{C_f}] \le \rho E[C_f]$, where $\rho$ is a constant that depends on $\sigma$ and $\varphi$.
\end{theorem}

\begin{proof}
The expected energy consumption of the solution found is $E[\hat{C_f}] = E[\sum_e f(s_e)]$. From Lemma~\ref{lm:bound_rounding-2} we have that $\hat{C_f} \le \varphi \sigma \hat{C_g}$ and, hence, $E[\hat{C_f}] \le \varphi \sigma E[\hat{C_g}]$. As it is shown in Lemma~\ref{lm:bound_rounding-1}, there is a constant $\lambda$ such that $E[\hat{C_g}] \le \lambda E[C_g^*]$.

To complete the proof, we observe from Theorem~\ref{thm:error_bound} that, for all $x$, $g(x)/\varphi \le f(x)$. Then the optimal fractional solution under $g(x)$ satisfies $C_g^*/\varphi \le C_f^*$. Given that the optimal fractional solution under $f(x)$ is always a lower bound of the optimal solution $C_f$, we also have that $C_g^* \le \varphi C_f^* \le \varphi C_f$. 
The last inequality comes from the fact that the optimal fractional solution is always a lower bound for the optimal solution. 
Combining all these results together, we can obtain that 
\begin{equation}
E[\hat{C_f}] \le \varphi \sigma E[\hat{C_g}] \le \varphi \sigma \lambda E[C_g^*] \le \varphi^2 \sigma \lambda E[C_f].
\end{equation}

Assume that $\rho = \varphi^2 \sigma \lambda$. As $\alpha$, $\varphi$ and $\lambda$ are constants, we conclude that $\rho$ will also be a constant. This completes the proof.
\end{proof}

This result can be directly applied to uniform demands, where each traffic demand requests a bandwidth $d_i = d$. In this case the total flow on each edge will be $d$ times that of the case with unit demands. This is the only difference between these two cases.

\begin{corollary}
A $\rho$-approximation to the optimal integral solution in expectation can be obtained for uniform demands.
\end{corollary}

For non-uniform demands, Andrews et al. \cite{Andrews_Fernandez-SS-2010} has derived an approximation result on the randomized rounding on paths. We borrow this result and combine it with Theorem~\ref{thm:error_bound} and Lemma~\ref{lm:bound_rounding-2}. Then the following theorem  holds.

\begin{theorem}
\label{thm:approx_ratio_nonuniform}
For non-uniform demands, we can achieve a $O(log^{\beta-1}d)$-approximation for the rate adaptive energy-efficient routing problem, where $d$ is the maximum traffic demand in $D$.
\end{theorem}

\section{Conclusion and Future Work}
\label{sec:con}

In this paper, we investigate the energy-efficient network routing problem with discrete cost functions. Our contributions are mainly on the following results. For the rate adaptive energy-saving problem, we formulate it as an integer program and prove that the general case of this problem is NP-hard, and even hard to approximate. For a special but practical case of this problem, we provide an efficient approximation algorithm which is based on a two-step relaxation and rounding process. We show by analysis that the error incurred in each step of our algorithm can be bounded. By using our algorithm we obtain a constant approximation to the optimal for unit and uniform demands and an $O(log^{\beta-1}d)$-approximation for non-uniform demands. 

Observing that the network traffic comes from the end devices connected to the network, we have that it is important to combine the optimization of placing end devices with the network routing optimization. For instance in a datacenter environment, a good placing of virtual machines can significantly reduce the network congestion as well as the energy consumption of network devices. As a natural extension to our routing problem, this new problem can be described by an additional description to the original model saying that we can also determine the sources and destinations of all the demands as well as the routing paths. Actually this turns out to be a combination of a quadratic assignment problem and a network routing problem. In general, this problem is much harder than the one we discussed in this paper. To find effective ways to approximate it is still an open problem.

\section*{Acknowledgement}

This research was supported in part by the National Natural Science Foundation of China grant 61020106002, 61161160566 and 61202059, and the Comunidad de Madrid grant S2009TIC-1692, Spanish MICINN grant TEC2011-29688-C02-01.





\bibliographystyle{elsarticle-num}
\bibliography{ref.bib}







\end{document}